\documentclass[journal]{IEEEtran}
\usepackage{epsfig,graphicx,subfigure,psfrag,amsmath,cases}
\usepackage{latexsym,amsmath,amssymb,epsfig,subfigure,algorithm,mathtools}
\usepackage{algorithmic}
\usepackage{color}
\usepackage{url}
\usepackage{scrtime}
\usepackage{cite}
\usepackage{subfigure}
\usepackage{bbding}
\usepackage{multicol}

\author{{Zhiqiang Wei, Derrick Wing Kwan Ng, and Jinhong Yuan\vspace{-10mm}}
\thanks{Zhiqiang Wei, Derrick Wing Kwan Ng, and Jinhong Yuan are with the School of Electrical
Engineering and Telecommunications, the University of New South Wales, Australia (email: zhiqiang.wei@student.unsw.edu.au; w.k.ng@unsw.edu.au; j.yuan@unsw.edu.au).
Derrick Wing
Kwan Ng is supported under Australian Research Councils Discovery Early
Career Researcher Award funding scheme (project number DE170100137).
This work was supported in part by the Australia Research Council (ARC)
Discovery Project DP160104566 and Linkage Project LP160100708.}}

\title{\vspace{-5mm}Joint Pilot and Payload Power Control for Uplink MIMO-NOMA with MRC-SIC Receivers}

\newtheorem{Thm}{Theorem}
\newtheorem{proof}{proof}

\newtheorem{T-Prob}{Transformed Problem}

\DeclareMathOperator{\maxo}{maximize}

\begin{document}
\maketitle
\begin{abstract}
This letter proposes a joint pilot and payload power allocation (JPA) scheme to mitigate the error propagation problem for uplink multiple-input multiple-output non-orthogonal multiple access (MIMO-NOMA) systems.
A base station equipped with a maximum ratio combining and successive interference cancellation (MRC-SIC) receiver is adopted for multiuser detection.
The average signal-to-interference-plus-noise ratio (ASINR) of each user during the MRC-SIC decoding is analyzed by taking into account the error propagation due to the channel estimation error.
Furthermore, the JPA design is formulated as a non-convex optimization problem to maximize the minimum weighted ASINR and is solved optimally with geometric programming.
%
%
Simulation results confirm the developed performance analysis and show that our proposed scheme can effectively alleviate the error propagation of MRC-SIC and enhance the detection performance, especially for users with moderate energy budgets.
%
\end{abstract}
\section{Introduction}
Non-orthogonal multiple access (NOMA) has recently been recognized as a promising multiple access solution to fulfill the stringent quality of service (QoS) requirements of the fifth-generation (5G) wireless networks, such as high spectral efficiency and massive connectivity\cite{wong2017key}.
The principle of NOMA is to exploit the power domain for multiuser multiplexing and to adopt the successive interference cancellation (SIC) decoding at receivers to mitigate the multiuser interference\cite{wong2017key}.
In recent years, downlink NOMA has been extensively studied in the literature and it has been shown that downlink NOMA can achieve a considerable performance gain over conventional orthogonal multiple access (OMA) schemes in terms of spectral efficiency and energy efficiency.

In fact, NOMA inherently exists in uplink communications, since the electromagnetic waves are naturally superimposed at a receiving base station (BS) and the implementation of SIC is more affordable for BSs than user terminals.
For instance, a simple back-off power control scheme was proposed for uplink NOMA\cite{Zhangtobepublished}, while an optimal resource allocation algorithm to maximize the system sum rate was developed in \cite{Al-Imari2014}.
The authors in \cite{Yang2016NOMA} proposed a general power control framework to guarantee the QoS in downlink and uplink NOMA.
Most recently, multiple-input multiple-output NOMA (MIMO-NOMA) systems are of more interests\cite{Ding2015a,Xu2017}.
In particular, maximum ratio combining with successive interference cancellation (MRC-SIC) is an appealing and practical reception technique for uplink MIMO-NOMA owing to its low computational complexity.

Despite the fruitful research conducted on NOMA, only payload power allocation and ideal SIC decoding are considered in most of existing works, e.g.\cite{Al-Imari2014,Wei2017}.
For both single-antenna and multiple-antenna systems, it is well-known that error propagation of SIC decoding limits the promised performance gain brought by NOMA.
In practice, the sources of error propagation are two-fold: one is the channel estimation error (CEE) and the other is the erroneous in data detection.
This letter focuses on tackling the former issue via exploiting the non-trivial trade-off between the pilot and payload power allocation for uplink MIMO-NOMA systems for a given total energy budget.
Specifically, a higher pilot power yields a better channel estimation but leads to a less payload power for data detection.
In the meantime, the reduced payload power would introduce a lower inter-user interference (IUI) for other users.
Therefore, jointly designing the pilot and payload power allocation is critical for mitigating the error propagation.

In this letter, to alleviate the error propagation in SIC, we propose a joint pilot and payload power allocation (JPA) scheme for uplink MIMO-NOMA with a MRC-SIC receiver based on a practical minimum mean square error (MMSE) channel estimator.
%
%
We analyze the average signal-to-interference-plus-noise ratio (SINR) of each user during the MRC-SIC decoding.
Furthermore, under a total energy budget constraint for each user, the JPA design is formulated as a non-convex optimization problem to maximize the minimum weighted average SINR (ASINR).
The globally optimal solution of the JPA design problem is obtained by geometric programming.
Simulation results demonstrate that the proposed JPA scheme is beneficial to mitigate the error propagation, which enhances the data detection performance, especially in the moderate energy budget regime.

\section{System Model}

\subsection{System Model}
\begin{figure}[t]
\centering
\includegraphics[width=2.0in]{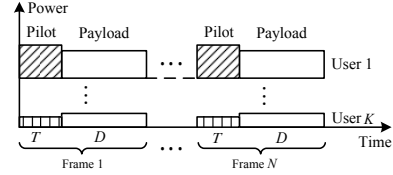}\vspace{-3mm}
\caption{An illustration of the frame structure of the uplink transmission.}\label{FrameStructure}\vspace{-6mm}
\end{figure}

We consider an uplink MIMO-NOMA communication system in a single-cell with a BS equipped with $M$ antennas serving $K$ single-antenna users.
All the $K$ users are allocated on the same frequency band.
Every user transmits multiple frames over multiple coherence time intervals (CTI) to the BS, where we assume that the duration of each frame is comparable to that of a CTI.
In particular, each frame consists of $T$ pilot symbols and $D$ data symbols consecutive in time, as shown in Figure \ref{FrameStructure}.
We assume that $T$ and $D$ are fixed as it is commonly implemented in practical systems for simplifying time synchronization.
Instead of considering symbol-level SIC as in most of existing works in NOMA \cite{Wei2017,Ding2015a}, we adopt the codeword-level SIC to exploit coding gain.
%
%
Note that, a codeword is usually much longer than the duration of a CTI and is spread over $N$ frames, which is an important scenario for time-varying channels with a short coherence time.

In frame $n$, the received signal at the BS during pilot transmission and data transmission are given by
\begin{equation}
\vspace{-2mm}
\hspace{-2mm}{{\mathbf{Y}}_{n}^{\mathrm{P}}} = {{\mathbf{H}}_n}{\mathbf{\Lambda T}} + {{\mathbf{Z}}_{n}^{\mathrm{P}}} \; \text{and}\;
{{\mathbf{Y}}_{n}^{\mathrm{D}}} = {{\mathbf{H}}_n}{\mathbf{B}}{{\mathbf{D}}_n} + {{\mathbf{Z}}_{n}^{\mathrm{D}}},
\end{equation}
respectively.
$\mathbf{T} \in \mathbb{C}^{K\times T}$ denotes the pilot matrix and ${\mathbf{D}_n} = \left[{\mathbf{d}_{n,1}},\ldots, \mathbf{d}_{n,K}\right]^{\mathrm{T}}\in \mathbb{C}^{K\times D}$ denotes the data matrix in frame $n$.
The diagonal matrices ${\mathbf{\Lambda}}$ and ${\mathbf{B}}$ are defined by ${\mathbf{\Lambda}} = \text{diag}\left\{\sqrt{\alpha_1}, \ldots, \sqrt{\alpha_K}\right\}$ and ${\mathbf{B}} = \text{diag}\left\{\sqrt{\beta_1}, \ldots, \sqrt{\beta_K}\right\}$, respectively, where $\alpha_k$ and $\beta_k$ denote the pilot and payload power of user $k$, respectively\footnote{The power allocation for pilot and payload are calculated centrally at the BS and to be distributed to all the users through some closed-loop power control scheme in control channels, e.g. \cite{Lee1996CDMA}.}.
We assume that normalized orthogonal pilots are assigned to all the users exclusively, i.e., ${\mathbf{T}}{\mathbf{T}}^{\mathrm{H}} = \mathbf{I}_K$ with $T \ge K$.
%
%
%
%
%
The matrices $\mathbf{Z}_{n}^{\mathrm{P}} \in \mathbb{C}^{M\times T}$ and $\mathbf{Z}_{n}^{\mathrm{D}} \in \mathbb{C}^{M\times D}$ denote the additive zero mean Gaussian noise with covariance matrix ${\sigma ^2}\mathbf{I}_M$ during training phase and data transmission phase in frame $n$, respectively.
The matrix $\mathbf{H}_n = \left[{{\mathbf{h}}_{n,1}},\ldots,{{\mathbf{h}}_{n,K}}\right] \in \mathbb{C}^{M\times K}$ contains the channels of all the users in frame $n$, where column $k$ denotes the channel vector of user $k$.
Rayleigh fading assumption is adopted in this letter, i.e.,
${{\mathbf{h}}_{n,k}} \sim \mathcal{CN}\left( \mathbf{0},\nu_k^2 \mathbf{I}_M \right)$, where $\mathcal{CN}\left( \mathbf{0},\nu_k^2 \mathbf{I}_M \right)$ denotes a circularly symmetric complex Gaussian distribution with zero mean and covariance matrix $\nu_k^2 \mathbf{I}_M$.
Scalar $\nu_k^2$ denotes the large scale fading of user $k$ capturing the effects of path loss and shadowing.
Since all the users are usually sufficiently separated apart compared to the wavelength, their channels are assumed to be independent with each other.
Therefore, their channel correlation matrix is given by a diagonal matrix ${{{\mathbf{R}}_{\mathbf{H}}}} = M\text{diag} \left\{\nu _1^2,\ldots,\nu _K^2\right\}$.
Without loss of generality, we assume that users are indexed in the descending order of large scale fading, i.e., $\nu _1^2 \ge \nu _2^2 \ge  \ldots  \ge \nu _K^2$.
In this letter, we define strong or weak user based on the large scale fading since it facilitates the characterization of the channel ordering statistically across the codeword, i.e., user 1 is the strongest user, while user $K$ is the weakest user.
As a result, the SIC decoding order is assumed to be the descending order of large scale fading, i.e., users $1,2,\ldots,K$ are decoded sequentially.

\section{Performance Analysis on ASINR}
In the $k$-th step of the MRC-SIC decoding, after cancelling the signals of the previous $k-1$ users, the post-processing signal of user $k$ in frame $n$ is given by
\vspace{-2mm}
\begin{align}\label{MRC-SIC}
\vspace{-9mm}
{{\mathbf{y}}_{n,k}^{\mathrm{T}}}
&= \underbrace{{\mathbf{\hat h}}_{n,k}^{\mathrm{H}}{{\mathbf{\hat h}}_{n,k}}\sqrt {{\beta _k}} {{\mathbf{d}}_{n,k}}}_{\text{desired signal}} + \underbrace{{\mathbf{\hat h}}_{n,k}^{\mathrm{H}}\sum\nolimits_{l = 1}^{k} {\boldsymbol{\varepsilon} _{n,l}} \sqrt {{\beta _l}} {{\mathbf{d}}_{n,l}}}_{\text{residual interference}} \notag\\[-1.5mm]
&+ \underbrace{{\mathbf{\hat h}}_{n,k}^{\mathrm{H}}\sum\nolimits_{l = k + 1}^K {{{\mathbf{h}}_{n,l}}} \sqrt {{\beta _l}} {{\mathbf{d}}_{n,l}}}_{\text{inter-user interference}}
+ \underbrace{{\mathbf{\hat h}}_{n,k}^{\mathrm{H}}{{\mathbf{Z}}_{n,d}}}_{\text{noise}},
\end{align}
where ${{\mathbf{y}}_{n,k}} \in \mathbb{C}^{T\times 1}$, ${{{\mathbf{\hat h}}}_{n,k}} \in \mathbb{C}^{M\times 1}$ denotes the MMSE channel estimates of user $k$ in frame $n$, and ${\boldsymbol{\varepsilon} _{n,k}} = {{\mathbf{h}}_{n,k}} - {{{\mathbf{\hat h}}}_{n,k}}$ denotes the corresponding CEE.
In \eqref{MRC-SIC}, we assume that the error propagation is only caused by the CEE but not affected by the erroneous in data detection.
It is a reasonable assumption if we can guarantee the ASINR of each user larger than a threshold to maintain the required bit-error-rate (BER) performance through the power control in the following.
To this end, we first define the instantaneous SINR of user $k$ in frame $n$ as:
\begin{equation}\label{SINR}
\vspace{-2mm}
{\mathrm{SIN}}{{\mathrm{R}}_{n,k}} = \frac{s_{n,k}}{G_{n,k} + Q_{n,k} + {\sigma ^2}}, \forall n,k,
\end{equation}
with $s_{n,k} = {{\mathbf{\hat h}}_{n,k}^{\mathrm{H}}{{{\mathbf{\hat h}}}_{n,k}}{\beta _k}}$, $G_{n,k} = \sum_{l = k + 1}^K {\frac{{{\mathbf{\hat h}}_{n,k}^{\mathrm{H}}{{\mathbf{h}}_{n,l}}{\mathbf{h}}_{n,l}^{\mathrm{H}}{{{\mathbf{\hat h}}}_{n,k}}}}{{{\mathbf{\hat h}}_{n,k}^{\mathrm{H}}{{{\mathbf{\hat h}}}_{n,k}}}}} {\beta _l}$, and $Q_{n,k} = \sum_{l = 1}^{k} {\frac{{{\mathbf{\hat h}}_{n,k}^{\mathrm{H}}\left( {{{\mathbf{h}}_{n,l}} - {{{\mathbf{\hat h}}}_{n,l}}} \right){{\left( {{{\mathbf{h}}_{n,l}} - {{{\mathbf{\hat h}}}_{n,l}}} \right)}^{\mathrm{H}}}{{{\mathbf{\hat h}}}_{n,k}}}}{{{\mathbf{\hat h}}_{n,k}^{\mathrm{H}}{{{\mathbf{\hat h}}}_{n,k}}}}} {\beta _l}$,
while the ASINR of user $k$ is defined by:
\begin{equation}\label{ASINR}
\vspace{-1mm}
{\overline{\mathrm{SINR}}_k} = {\rm{E}}\left\{ {\frac{{{s_{n,k}}}}{{{G_{n,k}} + {Q_{n,k}} + {\sigma ^2}}}} \right\}, \forall k,
\end{equation}
where $\mathrm{E}\left\{ \cdot \right\}$ denotes the expectation operation.
In fact, for codeword-level SIC, it is the ASINR rather than the instantaneous SINR that determines the detection performance\cite{GaoAESINR}.
Yet, for mathematical tractability, in the sequel, we adopt the lower bound of $\overline{\mathrm{SINR}}_k$ proposed in \cite{GaoAESINR} as
\begin{equation}\label{AESINR}
\vspace{-1mm}
\mathrm{ASINR}_k = \frac{{\mathrm{E}\left\{ {s_{n,k}}\right\}}}{\mathrm{E}\left\{ {{G_{n,k}}} \right\} + \mathrm{E}\left\{ {{Q_{n,k}}} \right\} + {\sigma ^2}} \le {\overline{\mathrm{SINR}}_k}, \forall k.
\end{equation}
Variable $s_{n,k}$ denotes the desired signal power of user $k$ in frame $n$, ${G_{n,k}}$ denotes the IUI power in the $k$-th step of MRC-SIC, and ${{Q_{n,k}}}$ denotes the residual interference power caused by CEE.
Clearly, $s_{n,k}$, ${G_{n,k}}$, and ${Q_{n,k}}$ are functions of pilot and payload power allocation.
Now, we express the closed-form of \eqref{AESINR} through the following theorem.

\begin{Thm}\label{theorem1}
For two independent random vectors ${\mathbf{x}}, {\mathbf{y}} \in \mathbb{C}^{M\times 1}$ with distribution of ${\mathbf{x}} \sim \mathcal{CN}\left( \mathbf{0},\sigma_x^2 \mathbf{I}_M \right)$, we define a scalar random variable ${\phi} = {{{\mathbf{y}}^{\mathrm{H}}{\mathbf{x}}}}/{{\left| {\mathbf{y}} \right|}}$.
Then, ${\phi}$ is independent with ${\mathbf{y}}$ and it is distributed as a complex Gaussian distribution with zero mean and variance of $\sigma_x^2$, i.e., ${\phi} \sim \mathcal{CN}\left( 0, \sigma_x^2\right)$.
\end{Thm}
\begin{proof}
It is clear that the random variable ${\phi}$ for a given ${\mathbf{y}}$ is complex Gaussian distributed, where its conditional mean and variance are given by
\begin{align}
\mathrm{E}\left\{ {{\phi}|{\mathbf{y}}} \right\} &= {{\mathbf{y}}^{\mathrm{H}}}\mathrm{E}\left\{ {\mathbf{x}} \right\}/{{\left| {\mathbf{y}} \right|}} = 0 \;\text{and}\notag\\[-0.5mm]
\mathrm{Var}\left\{ {{\phi}|{\mathbf{y}}} \right\} &= {{{\mathbf{y}}^{\mathrm{H}}\mathrm{E}\left\{ {{\mathbf{x}}{{\mathbf{x}}^{\mathrm{H}}}} \right\}{\mathbf{y}}}}/{{{{\left| {\mathbf{y}} \right|}^2}}} = \sigma_x^2,
\end{align}
respectively.
Since the conditional mean and variance of ${\phi}$ are uncorrelated with ${\mathbf{y}}$, hence ${\phi}$ is independent of ${\mathbf{y}}$ with zero mean and variance of $\sigma_x^2$.
This completes the proof.
\end{proof}

%

Now, following the MMSE channel estimation\cite{BigueshMMSE2006} and invoking Theorem 1, we can easily obtain $\mathrm{E}\left\{ {{s_{n,k}}} \right\} = {{\mathbf{A}}_k^{\mathrm{H}}}{{\mathbf{\Phi }}^{ - 1}}{{\mathbf{A}}_k}$ and $\mathrm{E}\left\{ {{G_{n,k}}} \right\} = \sum\nolimits_{l = k + 1}^K \nu_l^2 {\beta _l}$,
where ${\mathbf{\Phi }} = {{\mathbf{T}}^{\mathrm{H}}}{\mathbf{\Lambda}}{{\mathbf{R}}_{\mathbf{H}}}{\mathbf{\Lambda}}{\mathbf{T}} + {\sigma ^2}M{\mathbf{I}}_{T}$, ${{\mathbf{A}}_k} = {{\mathbf{T}}^{\mathrm{H}}}{\mathbf{\Lambda}}{\left\{ {{{\mathbf{R}}_{\mathbf{H}}}} \right\}_{:k}}$, and ${\left\{ {{{\mathbf{R}}_{\mathbf{H}}}} \right\}_{:k}}$ denotes the $k$-th column of ${{{\mathbf{R}}_{\mathbf{H}}}}$.
Then, based on the orthogonal principle of the MMSE channel estimation\cite{BigueshMMSE2006} and Theorem 1, we have $\mathrm{E}\left\{ {{Q_{n,k}}} \right\} = \sum\nolimits_{l = 1}^{k} \sigma_{l}^2 {\beta _l}$.
Therein, variable $\sigma_{k}^2$ is the variance of the CEE of user $k$, which is obtained based on equation (27) in \cite{BigueshMMSE2006} as $\sigma _k^2 = \nu_k^2 - \frac{1}{M} \mathbf{A}_k^{\mathrm{H}}\mathbf{\Phi}^{-1}\mathbf{A}_k$.

Substituting $\mathrm{E}\left\{ {{s_{n,k}}} \right\}$, $\mathrm{E}\left\{ {{G_{n,k}}} \right\}$, and $\mathrm{E}\left\{ {{Q_{n,k}}} \right\}$ into \eqref{AESINR} and using the matrix inverse lemma on ${{\mathbf{\Phi }}^{ - 1}}$, we have
\begin{equation}\label{AESINR2}
\mathrm{ASINR}_k = \frac{{M\left( {\nu _k^2 - \frac{{{\sigma ^2}\nu _k^2}}{{{\sigma ^2} + {\alpha _k}\nu _k^2}}} \right){\beta _k}}}{{\sum\nolimits_{l = k + 1}^K {\nu _l^2} {\beta _l} + \sum\nolimits_{l = 1}^k {\frac{{{\sigma ^2}\nu _l^2}}{{{\sigma ^2} + {\alpha _l}\nu _l^2}}} {\beta _l} + {\sigma ^2}}}.
\end{equation}
We note that ${\mathrm{ASIN}}{{\mathrm{R}}_k}$ increases with $\left\{{\alpha _1}\right.$, $\ldots$, ${\alpha _k}$, $\left.{\beta _k}\right\}$, but decreases with $\left\{{\beta _1}\right.$, $\ldots$, ${\beta _{k-1}}$, ${\beta _{k+1}}$, $\ldots$, $\left.{\beta _K}\right\}$.
In other words, there exists a non-trivial trade-off between the allocation of pilot power and payload power.
Indeed, a higher pilot power ${\alpha _k}$ and payload power ${\beta _k}$ of user $k$ result in a higher ${\mathrm{ASIN}}{{\mathrm{R}}_k}$, while a higher payload power of other users ${\beta _1}$, $\ldots$, ${\beta _{k-1}}$, ${\beta _{k+1}}$, $\ldots$, ${\beta _{K}}$ will introduce more IUI for user $k$.
In contrast, high pilot powers ${\alpha _1}$, $\ldots$, ${\alpha _{k-1}}$ are beneficial to increase ${\mathrm{ASIN}}{{\mathrm{R}}_k}$ since they can reduce the residual interference by improving the quality of channel estimation.

\section{Joint Pilot and Payload Power Allocation}
The JPA design can be formulated to maximize the minimum weighted ${{\mathrm{ASIN}}{{\mathrm{R}}_k}}$ as follows\footnote{The performance degradation due to the adopted lower bound ASINR in \eqref{PilotPayloadPowerAllocation} is generally limited, while the simulation result for verification is not included in this letter due to the page limit.}:
\begin{align}
\label{PilotPayloadPowerAllocation}
&\underset{\left\{{\alpha_1}, \ldots, {\alpha_K}\right\},\left\{{\beta_1}, \ldots, {\beta_K}\right\}}{\maxo} \;\;\underset{k}{\min} \;\;\{{c_k{\mathrm{ASIN}}{{\mathrm{R}}_k}} \} \notag\\
\mbox{s.t.}\;\;\;\;
&\mbox{C1: } \alpha_k T + \beta_k D \le E_{\mathrm{max}}, \forall k, \notag\\[-1mm]
&\mbox{C2: } \alpha_k \ge 0, \beta_k \ge 0, \forall k, \;\;
\mbox{C3: } {\mathrm{ASIN}}{{\mathrm{R}}_k} \ge \gamma, \forall k.
\end{align}
\par\noindent
The constants ${\mathbf{c}} = \left[c_1, \ldots, c_K\right]$ are predefined weights for all the $K$ users.
Constraint C1 limits the pilot power ${\alpha_k}$ and the payload power $\beta_k$ with the maximum energy budget $E_{\mathrm{max}}$ for each user.
Constraint C2 ensures the non-negativity of ${\alpha_k}$ and $\beta_k$.
Constraint C3 requires the ASINR of user $k$ to be larger than a given threshold $\gamma$ to guarantee the data detection performance during SIC decoding.
Note that since the message of each user is decoded only once at the BS for uplink NOMA, a SIC decoding constraint is not required as imposed for downlink NOMA\cite{Hanif2016}.

This max-min problem formulation aims to mitigate the error propagation of the MRC-SIC decoding, which is dominated by the user with the minimum ASINR.
Furthermore, since the error propagation caused by the users at the forefront of the MRC-SIC decoding process, e.g. user 1, affects the data
detection of remaining undecoded users, and thus affects the system performance more significantly than other users.
Therefore, we have $c_1 \le c_2, \ldots, \le c_K$ to assign
different priorities to users in maximizing their ASINRs.
The formulated problem in (9) is a non-convex problem, where $\alpha_k$ and $\beta_k$ are coupled with each other severely in ${\mathrm{ASIN}}{{\mathrm{R}}_k}$.
%
Defining new optimization variables ${t_k} = \frac{{{\sigma ^2}\nu _k^2}}{{{\sigma ^2} + {\alpha _k}\nu _k^2}}$, $\forall k$, the problem in (9) is equivalent to the following optimization problem \cite{ChiangGP}:
\begin{align} \label{PilotPayloadPowerAllocation2}
&\underset{\left\{{t_1}, \ldots, {t_K}\right\},\left\{{\beta_1}, \ldots, {\beta_K}\right\},\lambda}{\maxo} \;\;\lambda \notag\\
\mbox{s.t.}\;\;\;\;
&\mbox{C1: } {\sigma ^2}Tt_k^{ - 1} + D{\beta _k} \le {{{\sigma ^2}T}}/{{\nu _k^2}} + {E_{{\mathrm{max}}}}, \forall k, \notag\\[-1mm]
&\mbox{C2: } 0 < t_k \le {{\nu _k^2}}, \beta_k \ge 0, \forall k, \notag\\[-1mm]
&\mbox{C3: } \sum\nolimits_{l = k + 1}^K {\gamma \nu _l^2{\beta _l}\beta _k^{ - 1}}  + \sum\nolimits_{l = 1}^k {\gamma {t_l}{\beta _l}\beta _k^{ - 1}}  + \gamma {\sigma ^2}\beta _k^{ - 1} \notag\\[-1mm]
&+ M{t_k} \le M\nu _k^2, \forall k, \notag\\[-1mm]
&\mbox{C4: } \sum\nolimits_{l = k + 1}^K {\nu _l^2\lambda {\beta _l}\beta _k^{ - 1}}  + \sum\nolimits_{l = 1}^k {{t_l}\lambda {\beta _l}\beta _k^{ - 1}}  + {\sigma ^2}\lambda \beta _k^{ - 1} \notag\\[-1mm]
&+ M{c_k}{t_k} \le M{c_k}\nu _k^2, \forall k,
\end{align}
where $\lambda > 0$ is an auxiliary optimization variable.
We can easily observe that the objective function and the functions on the left side of constraints C1, C3, and C4 in (10) are all valid posynomial functions\cite[Chapter~4]{Boyd2004}.
Therefore, the reformulated problem in (10) is a standard geometric programming (GP) problem\footnote{Actually, the problem transformation in transfoming (9) to (10) is standard and can be found in \cite{ChiangGP}.
However, without the proposed steps to simplify the performance analysis, the GP transformation\cite{ChiangGP} cannot be directly applied to the considered problem.}, which can be solved efficiently by off-the-shelf numerical solvers such as CVX\cite{cvx}.

\section{Simulation Results}

We use simulations to verify the developed performance analysis and evaluate the performance of the proposed JPA scheme for both uncoded and coded systems.
Two baseline schemes are introduced for comparison, where the equal power allocation (EPA) scheme sets the equal pilot and payload power, i.e., $\alpha_k = \beta_k = \frac{E_{\mathrm{max}}}{T+D}$, but the payload power allocation (PPA) scheme only fixes the pilot power as $\alpha_k = \frac{E_{\mathrm{max}}}{T+D}$ and optimizes over the payload power $\beta_k$ subject to the same constraint set as in \eqref{PilotPayloadPowerAllocation2}.

In the simulations, we set $M=2$, $T = K = 4$, $D = 96$, ${\mathbf{c}} = \left[\frac{1}{8}, \frac{1}{8}, \frac{1}{4}, \frac{1}{2}\right]$, $\gamma = 5$ dB, $E_{\mathrm{max}} = 20$ J, and ${\sigma ^2} = -100$ dBm.
It is assumed that there are $T+D = 100$ symbols in a CTI.
All the $K$ users are uniformly distributed in a single cell with a cell radius of $400$ m.
The weight ${\mathbf{c}}$ is selected deliberately to alleviate the impact of the error propagation from the previous users during the MRC-SIC decoding, where the optimal weight selection will be considered in future work.
The 3GPP urban path loss model\cite{Access2010} is adopted and quadrature phase shift keying (QPSK) modulation is used for all the simulation cases.
For the coded systems, we adopt the standard turbo code as stated in the 3GPP technical specification\cite{3GPPReportTurboCode2016}.
We assume that one codeword is spread over $N=10$ coherence intervals, which results in a codeword length of $1920$ bits.
%

\subsection{Individual ASINR}
Figure \ref{AESINR_Compare} depicts the individual ASINR for the considered three schemes for uncoded systems.
We can observe that the simulation results match perfectly with the theoretical results in \eqref{AESINR2}.
%
%
Besides, it can be observed that the lowest ASINR achieved by the PPA scheme and our proposed JPA scheme both occur at $5$ dB for user 4, which is much higher than the minimum ASINR provided by the EPA scheme occurring at $2.6$ dB for user 3.
This is owing to the adopted max-min principle and constraint C3 in the proposed problem formulation in \eqref{PilotPayloadPowerAllocation}.
Nevertheless, it can be observed that our proposed scheme provides a $2$ dB higher ASINR than that of the PPA scheme for users 1, 2, and 3.
This is because our proposed scheme can utilize the energy more efficiently than that of the PPA scheme.
Moreover, our simulation results demonstrate that the optimal power allocation $\alpha^{*}_k$ and $\beta^{*}_k$ can satisfy the energy budget constraint C1 in \eqref{PilotPayloadPowerAllocation}.
%

Furthermore, the Jain's fairness index (JFI) of the weighted ASINR for the considered three schemes are given by $J_{\mathrm{EPA}} = 0.6174$,
$J_{\mathrm{PPA}} = 0.9436$, and
$J_{\mathrm{JPA}} = 0.9983$,
respectively.
The EPA scheme achieves the lowest JFI, while both the PPA and JPA schemes enjoy a high JFI since they are based on the max-min resource allocation in \eqref{PilotPayloadPowerAllocation}.
In addition, our proposed JPA scheme offers a slightly higher JFI than that of the PPA scheme due to its efficient utilization of energy.

\subsection{Individual BER}
Figure \ref{BER_Compare} illustrates the individual BER performance for uncoded and coded systems.
We can observe that the coded system offers much lower BERs than the uncoded system owing to the coding gain.
For the EPA scheme, user 4 endures a high BER as user 3 despite it posses a larger ASINR than user 3, as shown in Figure \ref{AESINR_Compare}.
This reveals the error propagation of the MRC-SIC decoding for the EPA scheme.
The PPA scheme can improve the BER performance for users 1, 2, and 3 compared to the EPA scheme, while it fails to relieve user 4 from high BER.
However, our proposed scheme always enjoys the lowest BER compared to the two baseline schemes for all the users, especially for coded systems.
In fact, our proposed scheme can mitigate the error propagation more efficiently compared to the PPA scheme by optimally balancing the pilot and payload power.
It is worth to note that, with constraint C3 in \eqref{PilotPayloadPowerAllocation}, our proposed JPA scheme can guarantee the BER of all the users to be smaller than $10^{-3}$ for the coded systems, which validates the assumption about the sources of the error propagation in \eqref{MRC-SIC}.

\begin{figure}[t]
\vspace{-5mm}
\begin{minipage}{0.5\textwidth}
\hspace{-6mm}
\centering
\includegraphics[width=3in]{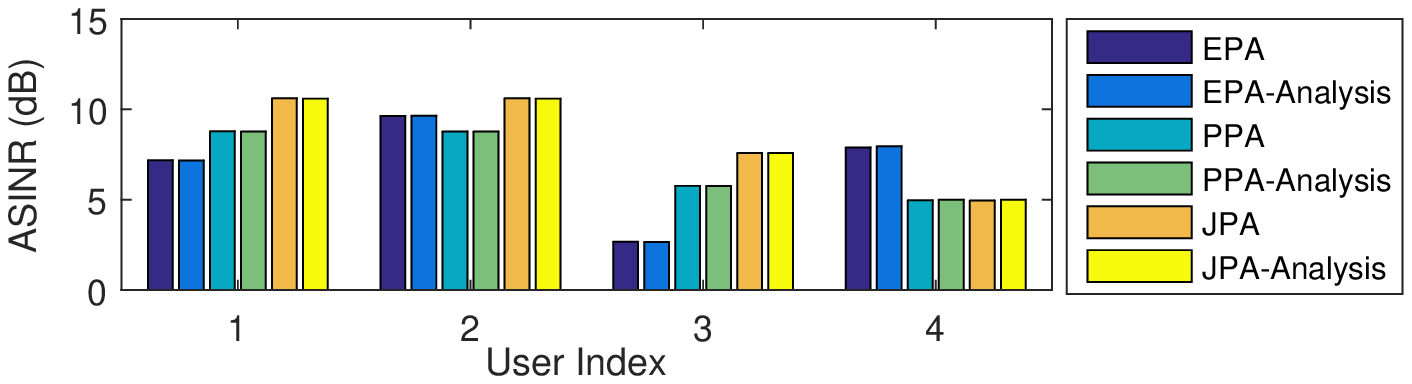}\vspace{-3mm}
\caption{Individual ASINR of uplink MIMO-NOMA with a MRC-SIC receiver.}
\label{AESINR_Compare}
\end{minipage}
\begin{minipage}{0.5\textwidth}
\centering\hspace{-5mm}
\includegraphics[width=3in]{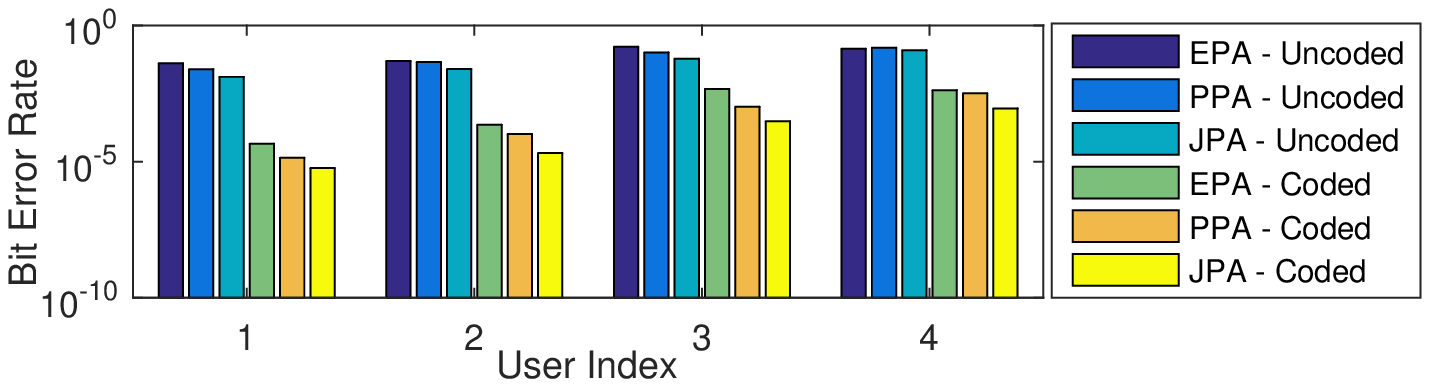}\vspace{-3mm}
\caption{Individual BER of uplink MIMO-NOMA with a MRC-SIC receiver.}\vspace{-5mm}
\label{BER_Compare}
\end{minipage}

\end{figure}

\subsection{BER versus Energy Budget}
For coded systems, Figure \ref{BER_SNR_Coded} shows the BER performance of our proposed scheme over the PPA scheme versus the energy budget $E_{\mathrm{max}}$.
%
%
Note that we set $\text{BER} = 0.5$ if the optimization problem in \eqref{PilotPayloadPowerAllocation} is infeasible to account the penalty of failure.
We can observe that our proposed scheme offers a much lower BER than that of the PPA scheme for all four users.
Interestingly, the BER performance gain is considerable in the moderate $E_{\mathrm{max}}$ regime, while it is marginal in the high $E_{\mathrm{max}}$ regime.
In fact, in the high $E_{\mathrm{max}}$ regime, the residual interference $Q_{n,k}$ vanishes owing to the high channel estimation accuracy.
%
%
Therefore, our proposed scheme can only offer diminishing gains in alleviating the impact of error propagation for further reducing the BER.
With the moderate $E_{\mathrm{max}}$, our proposed scheme can substantially improve the channel estimation, which can mitigate the residual interference during MRC-SIC decoding, and thus reduce the BER effectively.
In addition, it can be observed that an error floor for both
schemes appears at the BER region ranging from $10^{-2}$ to $10^{-5}$.
This early error floor is due to the joint effect of IUI, CEE, and the error propagation of the MRC-SIC decoding.
Note that an iterative receiver\cite{Xu2017} can be employed to lower the error floor level, which will be considered in our future work.

\section{Conclusion}
In this letter, a joint pilot and payload power control scheme was proposed for uplink MIMO-NOMA systems with MRC-SIC receivers to mitigate the error propagation problem.
By taking into account the CEE, we analyzed the ASINR during the MRC-SIC decoding.
The JPA design was formulated as a non-convex optimization problem for maximizing the minimum weighted ASINR and was solved by geometric programming.
Simulation results verified our analysis and demonstrated that our proposed scheme is effective in mitigating the error propagation in SIC which enhances the BER performance, especially in the moderate energy budget regime.

\begin{figure}[t!]
\centering\vspace{-5mm}
\includegraphics[width=3.3in]{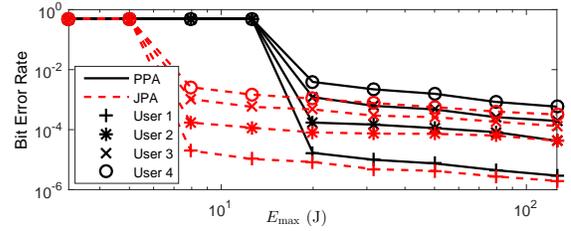}\vspace{-3mm}
\caption{BER performance versus energy budget $E_{\max}$ of uplink MIMO-NOMA with a MRC-SIC receiver.}\vspace{-5mm}
\label{BER_SNR_Coded}
\end{figure}


\begin{thebibliography}{10}
\providecommand{\url}[1]{#1}
\csname url@samestyle\endcsname
\providecommand{\newblock}{\relax}
\providecommand{\bibinfo}[2]{#2}
\providecommand{\BIBentrySTDinterwordspacing}{\spaceskip=0pt\relax}
\providecommand{\BIBentryALTinterwordstretchfactor}{4}
\providecommand{\BIBentryALTinterwordspacing}{\spaceskip=\fontdimen2\font plus
\BIBentryALTinterwordstretchfactor\fontdimen3\font minus
  \fontdimen4\font\relax}
\providecommand{\BIBforeignlanguage}[2]{{%
\expandafter\ifx\csname l@#1\endcsname\relax
\typeout{** WARNING: IEEEtran.bst: No hyphenation pattern has been}%
\typeout{** loaded for the language `#1'. Using the pattern for}%
\typeout{** the default language instead.}%
\else
\language=\csname l@#1\endcsname
\fi
#2}}
\providecommand{\BIBdecl}{\relax}
\BIBdecl

\bibitem{wong2017key}
V.~W. Wong, R.~Schober, D.~W.~K. Ng, and L.-C. Wang, \emph{Key Technologies for
  {5G} Wireless Systems}.\hskip 1em plus 0.5em minus 0.4em\relax Cambridge
  University Press, 2017.

\bibitem{Zhangtobepublished}
N.~Zhang, J.~Wang, G.~Kang, and Y.~Liu, ``Uplink non-orthogonal multiple access
  in {5G} systems,'' \emph{IEEE Commun. Lett.}, vol.~20, no.~3, pp. 458--461,
  Mar. 2016.

\bibitem{Al-Imari2014}
M.~Al-Imari, P.~Xiao, M.~A. Imran, and R.~Tafazolli, ``Uplink non-orthogonal
  multiple access for {5G} wireless networks,'' in \emph{Proc. IEEE Intern.
  Sympos. on Wireless Commun. Systems}, Aug. 2014, pp. 781--785.

\bibitem{Yang2016NOMA}
Z.~Yang, Z.~Ding, P.~Fan, and N.~Al-Dhahir, ``A general power allocation scheme
  to guarantee quality of service in downlink and uplink {NOMA} systems,''
  \emph{IEEE Trans. Wireless Commun.}, vol.~15, no.~11, pp. 7244--7257, Nov.
  2016.

\bibitem{Ding2015a}
Z.~Ding, R.~Schober, and H.~V. Poor, ``A general {MIMO} framework for {NOMA}
  downlink and uplink transmission based on signal alignment,'' \emph{IEEE
  Trans. Wireless Commun.}, vol.~15, no.~6, pp. 4438--4454, Jun. 2016.

\bibitem{Xu2017}
C.~Xu, Y.~Hu, C.~Liang, J.~Ma, and L.~Ping, ``Massive {MIMO}, non-orthogonal
  multiple access and interleave division multiple access,'' \emph{IEEE
  Access}, vol.~5, pp. 14\,728--14\,748, Jul. 2017.

\bibitem{Wei2017}
Z.~Wei, D.~W.~K. Ng, J.~Yuan, and H.~M. Wang, ``Optimal resource allocation for
  power-efficient {MC-NOMA} with imperfect channel state information,''
  \emph{IEEE Trans. Commun.}, vol.~PP, no.~99, pp. 1--1, May 2017.

\bibitem{Lee1996CDMA}
C.-C. Lee and R.~Steele, ``Closed-loop power control in {CDMA} systems,''
  \emph{IEE Proc. - Commun.}, vol. 143, no.~4, pp. 231--239, Aug. 1996.

\bibitem{GaoAESINR}
F.~Gao, R.~Zhang, and Y.~C. Liang, ``Optimal channel estimation and training
  design for two-way relay networks,'' \emph{IEEE Trans. Commun.}, vol.~57,
  no.~10, pp. 3024--3033, Oct. 2009.

\bibitem{BigueshMMSE2006}
M.~Biguesh and A.~B. Gershman, ``Training-based {MIMO} channel estimation: a
  study of estimator tradeoffs and optimal training signals,'' \emph{IEEE
  Trans. Signal Process.}, vol.~54, no.~3, pp. 884--893, Mar. 2006.

\bibitem{Hanif2016}
M.~F. Hanif, Z.~Ding, T.~Ratnarajah, and G.~K. Karagiannidis, ``A
  minorization-maximization method for optimizing sum rate in the downlink of
  non-orthogonal multiple access systems,'' \emph{IEEE Trans. Signal Process.},
  vol.~64, no.~1, pp. 76--88, Jan. 2016.

\bibitem{ChiangGP}
M.~Chiang, C.~W. Tan, D.~P. Palomar, D.~O'neill, and D.~Julian, ``Power control
  by geometric programming,'' \emph{IEEE Trans. Wireless Commun.}, vol.~6,
  no.~7, pp. 2640--2651, Jul. 2007.

\bibitem{Boyd2004}
S.~Boyd and L.~Vandenberghe, \emph{Convex optimization}.\hskip 1em plus 0.5em
  minus 0.4em\relax Cambridge university press, 2004.

\bibitem{cvx}
\BIBentryALTinterwordspacing
M.~Grant and S.~Boyd, ``{CVX}: Matlab software for disciplined convex
  programming, version 2.1.'' [Online]. Available: \url{http://cvxr.com/cvx}
\BIBentrySTDinterwordspacing

\bibitem{Access2010}
``Evolved universal terrestrial radio access: Further advancements for {E-UTRA}
  physical layer aspects,'' 3GPP TR 36.814, Tech. Rep., 2010.

\bibitem{3GPPReportTurboCode2016}
\BIBentryALTinterwordspacing
``Multiplexing and channel coding ({FDD}) (release 14),'' 3GPP TS 25.212, Tech.
  Rep., Dec. 2016. [Online]. Available:
  \url{http://www.3gpp.org/DynaReport/25212.htm}
\BIBentrySTDinterwordspacing
\end{thebibliography}

\end{document}